\documentclass{article}
\usepackage[plain]{fullpage}

\usepackage{amsmath}
\usepackage{amsthm}
\usepackage{amssymb}
\usepackage{leftidx}

\newtheorem{theorem}{Theorem}
\newtheorem{lemma}[theorem]{Lemma}
\newtheorem{proposition}[theorem]{Proposition}
\newtheorem{corollary}[theorem]{Corollary}
\newtheorem{fact}[theorem]{Fact}

\newtheorem{claim}[theorem]{Claim}
\newtheorem{definition}[theorem]{Definition}

\def\F{{\mathbb F}}

\newcommand{\ket}[1]{{|{#1}\rangle}}

\newcommand{\poly}{\mathrm{poly}}

\newcommand{\hparam}{\eta}

\newcommand{\HSP}{\mbox{\rmfamily\textsc{HSP}}}

\newcommand{\HPGP}{\mbox{\rmfamily\textsc{HPGP}}}

\newcommand{\SDE}{\mbox{\rmfamily\textsc{SDE}}}
\newcommand{\NP}{\mbox{\rmfamily\textsc{NP}}}
\newcommand{\suppress}[1]{}

\begin{document}


\title{Solving systems of diagonal polynomial equations 
over finite fields}

\author{G\'abor Ivanyos
\thanks{Institute for Computer Science and Control, Hungarian 
Academy of Sciences,
Kende u. 13-17, 
1111 Budapest, Hungary
{Email: \tt Gabor.Ivanyos@sztaki.mta.hu}}
\and
Miklos Santha
\thanks{CNRS, LIAFA, Universit\'e Paris Diderot 75205 Paris, France
and
Centre for Quantum Technologies, National University of Singapore, 
Singapore 117543
{Email: \tt miklos.santha@gmail.com}}}

\date{}
\maketitle

\begin{abstract}
We present an algorithm to solve a system of diagonal polynomial equations 
over finite fields when the number of variables is greater than some 
fixed polynomial of the number of equations whose degree depends only 
on the degree of the polynomial equations. 
Our algorithm works in time polynomial in the number of equations and 
the logarithm of the size of the field, whenever the degree of the polynomial 
equations is constant. As a consequence we design polynomial time quantum 
algorithms for two algebraic hidden structure problems: 
for the hidden subgroup problem in certain semidirect product 
$p$-groups of constant nilpotency class, and for the multi-dimensional 
univariate hidden polynomial graph problem when the degree of the 
polynomials is constant\footnote{This
is the accepted version of a manuscript to appear in
Theoretical Computer Science, see
{\tt http://dx.doi.org/10.1016/j.tcs.2016.04.045}.
\\
\copyright 2016. Made available under the CC-BY-NC-ND 4.0 license
{\tt http://creativecommons.org/licenses/by-nc-nd/4.0/}.
\\
An extended abstract reporting
on preliminary versions of the results 
has appeared in \cite{IS-FAW}.}.
\\~\\

\end{abstract}



\section{Introduction}
Finding small solutions in some well defined sense for a system of integer linear equations
is an important, well studied, and computationally hard problem. {\em Subset Sum}, 
which asks the solvability of a single equation in the binary domain is
one of Karp's
original 21 $\NP$-complete problems~\cite{Karp}.

The guarantees of many lattice based cryptographic systems come from the average case hardness
of {\em Short Integer Solution}, dating back to Ajtai's breakthrough work~\cite{Ajtai},
where we try to find short nonzero vectors in a random integer lattice.
Indeed, this problem has a remarkable worst case versus average case hardness property: solving
it on the average is at least as hard as solving various lattice problems in the worst case,
such as the decision version of the shortest vector problem, and finding short linearly independent vectors.

Turning back to binary solutions, deciding if there exists a 
nontrivial zero-one  solution of the system of linear equations
\begin{equation}
\begin{matrix}
a_{11}y_1 + \ldots + a_{1n}y_n &=& 0\\
\vdots & \vdots & \vdots \\
a_{m1}y_1 + \ldots + a_{mn}y_n &=& 0
\end{matrix}
\label{degreeone}
\end{equation}
in the finite field $\F_q$, where $q$ is a power of
some prime number $p$, is easy when $q=p=2$. However,
by modifying the standard reduction of {\em Satisfiability} to {\em Subset Sum}~\cite{Sipser}
it can be shown that it is an $\NP$-hard problem for $q \geq 3$.

The system~(\ref{degreeone}) is equivalent to the system of equations
\begin{equation}
\begin{matrix}
a_{11}x_1^{q-1} + \ldots + a_{1n}x_n^{q-1} &=& 0\\
\vdots & \vdots & \vdots \\
a_{m1}x_1^{q-1} + \ldots + a_{mn}x_n^{q-1} &=& 0
\end{matrix}
\label{degreeq-1}
\end{equation}
where we look for a nontrivial solution in the whole $\F_q^n$.

In this paper we will consider finding a nonzero solution for  a system of diagonal polynomial
equations similar to~(\ref{degreeq-1}), but where
more generally, the variables are raised to some power $d\geq 2$. We state
formally this problem.
\begin{definition}\label{def:SDE}
{\rm
The {\em System of Diagonal Equations} problem $\SDE$ is parametrized by a finite field
$\F_q$ and
three positive integers $n,m$ and $d$.\\
\vbox{\begin{quote}
$ \SDE(\F_q, n, m, d)$\\
{\em Input:} 
A system of polynomial equations over $\F_q$:
\begin{equation}
\begin{matrix}
a_{11}x_1^d + \ldots + a_{1n}x_n^d &=& 0\\
\vdots & \vdots & \vdots \\
a_{m1}x_1^d + \ldots + a_{mn}x_n^d &=& 0
\end{matrix}
\label{degreed}
\end{equation}
{\em Output:} A nonzero solution $(x_1, \ldots, x_n) \ne {\overrightarrow 0}$.
\end{quote}
}
}
\end{definition}
Here $\overrightarrow 0$ stands for the zero vector of length $n$. 
(We will use this notation
where we want to stress the distinction between the zero element of a field 
and the zero vector of a vector space.)

For $j=1, \ldots, n$, let us denote by $v_j$ the column vector 
$(a_{1j},\ldots,a_{mj})^T \in
\F_q^m$. Then
the system of equations~(\ref{degreed})  is the same as
\begin{equation}
\sum_{j=1}^n x_j^dv_j={\overrightarrow 0}.
\label{eq:vector}
\end{equation}
That is, solving $ \SDE(\F_q, n, m, d)$ is equivalent to 
the task of representing the zero vector as
a nontrivial linear combination of a subset of
$\{v_1,\ldots,v_n\}$
with $d$th power coefficients. We present our algorithm actually as solving this vector problem.
The special case $d=q-1$ is the vector zero sum problem
where the goal is to find a non-empty subset of the given vectors
with zero sum.

Under which conditions can we be sure that for system~(\ref{degreed}) there exists a nonzero solution?
The elegant result of Chevalley~\cite{Chevalley} and Warning~\cite{Warning} states that the number of solutions of a
general (not necessary diagonal) system of 
polynomial equations
is a multiple of the characteristic $p$ of $\F_q$, whenever
the number of variables is greater than the sum of the degrees of the polynomials.
For diagonal systems (\ref{degreed}) this means that 
when
$n>dm$, the existence of  a nonzero solution is assured. 

In general little is known about the complexity of finding another
solution, given a solution of a system which satisfies the Chevalley-Warning condition. 
When $q=2$, Papadimitriou has shown~\cite{Papa} that
this problem is in the complexity class Polynomial Parity Argument (PPA), 
the class of NP search 
problems where the existence of the solution is guaranteed by the fact that
in every finite graph the number of vertices with odd degree is even.
This implies that it cannot be NP-hard unless NP = co-NP. 
It is also unlikely that the problem is in P since Alon has shown~\cite{Alon} that this would 
imply that there are no one-way permutations.

Let us come back to our special system of equations~(\ref{degreed}). 
In the case $m=1$, a nonzero solution can be found in polynomial time for a single equation which satisfies
the Chevalley condition due to the remarkable work
of van de Woestijne~\cite{vdW} where he proves the following. 
\begin{fact}\label{fact:woestijne}
In deterministic polynomial time in $d$ and $\log q$ we can find a nontrivial solution for
\begin{equation*}
a_1 x_1^d + \ldots  + a_{d+1} x_{d+1}^d =0.  
\label{single-hom}
\end{equation*}
\end{fact}
In the case of more than one equation we don't know how to find a nonzero solution 
for system~(\ref{degreed}) under just the Chevalley condition. However, if we relax the problem,
and take much more variables than are required for the existence of a nonzero solution, we are able to give a
polynomial time solution.
Using van de Woestijne's result for the one dimensional case,
a simple recursion 
based on reducing one big system with $m$
 equations into $d+1$ subsystems with $m-1$ equations
shows that if $n \geq (d+1)^m$ then 
$ \SDE(\F_q, n, m, d)$ can be solved in deterministic polynomial time in $n$ and $\log q$.
The time complexity of this algorithm is therefore polynomial for any fixed $m$.
The case when $d$ is fixed
and $m$ grows appears to be more difficult. 
To our knowledge,
the only existing result in this direction is the case
$d=2$ for which it was shown in the paper~\cite{ISS} by the authors
and Sanselme
that there exists a (randomized) algorithm that,
when $n = \Omega(m^2)$,
solves 
$ \SDE(\F_q, n, m, d)$ in polynomial time in $n$ and $\log q$.
In the main result of this paper we generalize this result 
by showing, for every constant $d$,
the existence of 
a deterministic algorithm 
that, for every $n$ larger than some polynomial function of $m$,
solves $ \SDE(\F_q, n, m, d)$ in polynomial time in $n$ and $\log q$.
\begin{theorem}\label{thm:sde}
Let $d$ be constant. 
For $n > d^{d^2 \log d} (m+1)^{d \log d}$, the problem
$ \SDE(\F_q, n, m, d)$
can be solved in time polynomial in $n$ and $\log q$.
\end{theorem}
The large number of variables that 
makes a polynomial time solution possible, unfortunately also makes our algorithm
most probably irrelevant for cryptographic applications. Nonetheless, it turns out
that the algorithm is widely
applicable in quantum computing for solving efficiently various algebraic hidden structure problems.
We now explain this connection. 

Simply speaking, in a hidden structure problem we have to find some hidden object
related to some explicitly given algebraic structure $A$.
We have access to an oracle input, which is an unknown member $f$ of a family of black-box functions 
which map $A$ to some finite set $S$. The task is to identify
the hidden object solely from the information one can obtain
by querying the oracle $f$. This means that the only useful information we
can obtain is the structure of the { level sets} $f^{-1}(s)=\{ a
\in A: f(a)=s\}$, $s\in S$, that is, we can only determine whether two
elements in $A$ are mapped to the same value or not. 
In these problems we say that the input $f$ {\em hides} the hidden structure, the output of the problem.
We define now the two 
problems for which we can apply our algorithm for $\SDE$.
\begin{definition}\label{def:hidden}
{\rm
The {\em hidden subgroup problem} $\HSP$ is parametrized by a finite group $G$ and
a family ${\cal H}$ of subgroups of $G$.\\
\vbox{\begin{quote}
\HSP$(G, {\cal H})$ \\
{\em Oracle input:} 
A function $f$ from $G$ to some finite set $S$. \\
{\em Promise:} For some subgroup $H \in {\cal H}$, we have
$$f(x) = f(y) \Longleftrightarrow Hx = Hy.$$
{\em Output:} $H$.
\end{quote}}\\
The {\em hidden polynomial graph problem} $\HPGP$ is parametrized 
by a finite field $\F_q$ and three positive integers $n, m$ and $d$.\\
\vbox{\begin{quote}
\HPGP$(\F_q, n,m,d)$.\\
{\em Oracle input:} 
A function $f$ from $\F_q^n \times \F_q^m$ to a finite set $S$. \\
{\em Promise:}
For some 
$Q : \F_q^n \rightarrow \F_q^m$, where
$Q(x) = (Q_1(x), \ldots , Q_m(x))$, and $Q_i(x)$ is an $n$-variate degree $d$ polynomial over $\F_q$
with zero constant term,
we have
$$f(x,y) = f(x', y') \Longleftrightarrow y - Q(x) = y' - Q(x').$$
{\em Output:} $Q$.
\end{quote}}}
\end{definition}

While no classical algorithm can solve the $\HSP$ with polynomial query complexity 
even if the group $G$ is abelian, 
one of the most powerful results of quantum computing is that it can be solved by a 
polynomial time quantum algorithm for any abelian $G$. Shor's factorization and discrete 
logarithm finding algorithms~\cite{sho97}, and Kitaev's algorithm~\cite{Kitaev} for 
the abelian stabilizer problem are all special cases of this general solution.

Extending the quantum solution of the
abelian HSP to non abelian groups is an active research area
since these instances include several algorithmically important problems.
For example, efficient solutions for the dihedral and
the symmetric group would imply efficient solutions, respectively, for several lattice
problems~\cite{Regev} and for graph isomorphism.  While the non abelian
$\HSP$ has been solved efficiently by quantum algorithms in various 
groups~\cite{BCvD,DMR10,FIMSS03,gsvv01,hrt03,Kuperberg,MRRS04},
finding a general solution seems totally elusive.

An extension in a seemingly different (not "group theoretical")
framework was proposed by Childs,
Schulman and Vazirani~\cite{CSV} who considered the problem where the hidden object is a
polynomial. To recover it we have at our disposal an oracle whose
level sets coincide with the level sets of the polynomial. Childs et
al.~\cite{CSV} showed that the quantum query complexity of this
problem is polynomial in the logarithm of the field size when
the degree and the number of variables are constant.
The first time-efficient quantum algorithm was given by the authors
with Decker and Wocjan~\cite{DISW} for the case of multivariate quadratic
polynomials over fields of constant characteristic. 

The hidden polynomial graph problem $\HPGP$ was defined
in~\cite{DDW09} by Decker, Draisma and Wocjan.
Here the hidden object is again a polynomial, but the oracle is more powerful than in~\cite{CSV} because
it can also be queried on the graphs that are defined by the
polynomial functions.   They obtained a polynomial time quantum
algorithm that correctly identifies the hidden polynomial when the
degree and the number of variables are considered to be constant.
In~\cite{DISW}, this result was extended to polynomials of 
constant degree in a framework that reveals relationship
 to the hidden subgroup problem.
The version of the $\HPGP$ we define here is more general
than the one considered in~\cite{DDW09} in the sense that we are dealing not only with a single polynomial
but with a vector of several polynomials. The restriction on the constant terms of the polynomials
is due to the fact
that level sets of two polynomials are the same if they differ only in their constant terms, and therefore the value of the 
constant term can not be recovered.

It will be convenient for us to consider a slight variant of the hidden polynomial graph problem which
we denote by $\HPGP'$. The only difference between the two problems is that in the case of $\HPGP'$
the input is not given by an oracle function but by the ability to access random {\em level set states},
which are quantum states of the form
\begin{equation}
\label{eq:levstate}
\sum_{x \in \F_q^n}\ket{x}\ket{u+Q(x)},
\end{equation}
where $u$ is a random element of $\F_q^m$. Given an oracle input $f$ for $\HPGP$, a simple and efficient
quantum algorithm can create such a random coset state. Therefore an efficient quantum algorithm for 
$\HPGP'$ immediately provides an efficient quantum algorithm for $\HPGP$.

In~\cite{DHIS} the authors with Decker and H{\o}yer showed
that $\HPGP'(\F_q, 1, m, d)$ is solvable in quantum polynomial time
when $d$ and $m$ are both constant. Part of the quantum algorithm repeatedly solved instances
of $\SDE(\F_q, n, m, d)$ under such conditions. We present here a modification of this method
which works in polynomial time even if $m$ is not constant. For simplicity,
here we restrict ourselves to prime fields. This will be still sufficient
for application to a hidden subgroup problem.
\begin{theorem}\label{thm:hpgp_to_sde}
Let $d$ be constant and $p$ be a prime. If $\SDE(\F_p, n, m, d)$ is solvable in 
(randomized) polynomial time for some $n$,
then $\HPGP'(\F_p, 1, m, d)$  is 
solvable in quantum polynomial time.
\end{theorem}
Using Theorem~\ref{thm:sde} it is possible to dispense in the result 
of the authors with Decker and H{\o}yer~\cite{DHIS} with the assumption 
that $m$ is constant.
\begin{corollary}\label{cor:hpgp}
If $d$ is constant then $\HPGP'(\F_p, 1, m, d)$  is 
solvable in quantum polynomial time.
\end{corollary}
Bacon, Childs and van Dam in~\cite{BCvD} have considered the $\HSP$ in $p$-groups of the form 
$G=\F_p\ltimes \F_p^m$ when the hidden subgroup belongs to the family ${\cal H}$ 
of subgroups of order $p$ which are not subgroups of the normal subgroup $0 \times \F_p^m$.
They have found an efficient quantum algorithm for such groups as long as $m$ is constant.
In \cite{DISW}, 
based on arguments from
\cite{BCvD} the authors with Decker and H{\o}yer
sketched how the $\HSP(G, {\cal H})$ can be translated
into a hidden polynomial graph problem. For the sake of completeness we state 
here and prove the exact statement about such a reduction.
\begin{proposition}\label{prop:hsp_to_hpgp}
Let $d$ be the nilpotency class of a group $G$ of the form
$\F_p\ltimes \F_p^m$. There is a polynomial time quantum algorithm which reduces
$\HSP(G, {\cal H})$ to $\HPGP'(\F_p, 1, m, d)$.
\end{proposition}
Putting together Corollary~\ref{cor:hpgp} and Proposition~\ref{prop:hsp_to_hpgp},
it is also possible to get rid of the assumption that $m$ is constant in the result of~\cite{BCvD}.
\begin{corollary}
If the nilpotency class of the group $G$ of the form 
$\F_p\ltimes \F_p^m$ is constant then 
$\HSP(G, {\cal H})$ can be solved in quantum polynomial time.
\end{corollary}
We illuminate the main ideas of the proof of Theorem~\ref{thm:sde}
by showing special cases of weaker (randomized) versions 
for $d=2,3$ in Section 2.
Actually, randomization in these algorithms is only required to obtain
quadratic and cubic nonresidues in $\F_q$. We remark that assuming the
Extended Riemann hypothesis, such nonresidues can be found even
deterministically in time polynomial in $\log q$, see~\cite{Bach}.  
The proof of Theorem~\ref{thm:sde} will be given in Section 3. There
we also show how necessity of having nonresidues can be got around.
Finally the
proof of Proposition~\ref{prop:hsp_to_hpgp} will be given in Section 4, and the proof 
of Theorem~\ref{thm:hpgp_to_sde} in Section 5.

\section{Warm-up: the quadratic and cubic cases}

\subsection{The quadratic case}
\begin{proposition}\label{prop:quadratic}
The problem $\SDE(\F_q, (m+1)^2, m, 2)$ can be solved by a randomized
algorithm in time polynomial in $\log q$ and $m$.
\end{proposition}
\begin{proof}
We assume that $p>2$ 
and that we have a non-square $\zeta$ in $\F_q$ at hand.
Such an element can be efficiently found by a random choice. Actually,
this is the only point of our algorithm where randomization is used.
Assuming
ERH, even a deterministic polynomial time method exists for finding
a non-square. Also, as we will see in Section~3, one can even get
around the necessity of nonresidues. As we present this proof and that
for the cubic case for showing the main lines of our general algorithm,
we do not address this issue here.

Our input is a set $V$ of $(m+1)^2$ vectors in $\F_q^m$, and we want to represent the
zero vector as a nontrivial linear combination of some vectors from $V$ where all the coefficients are squares.
The construction is based on the following. Pick any
$m+1$ vectors $v_1,\ldots,v_{m+1}$ from $V$. Since they are linearly dependent, it is easy
to represent the zero
vector as a proper linear combination $\sum_{i=1}^{m+1}\alpha_iv_i = 0$.
Let $J_1=\{i:\alpha_i^{\frac{q-1}{2}}=1\}$ and
$J_2=\{i:\alpha_i^{\frac{q-1}{2}}=-1\}$. 
Using $\zeta$, we can find 
in deterministic polynomial
time in $\log q$ by the Shanks-Tonelli algorithm~\cite{Shanks} field elements
$\beta_i$ such that
$\alpha_i=\beta_i^2$ for $i\in J_1$ and
$\alpha_i=\beta_i^2\zeta$ for $i\in J_2$.
Let $w_{1}=\sum_{i\in J_1} \beta_i^2v_i$ and
$w_{2}=\sum_{i\in J_2}\beta_i^2v_i$. Then
$w_{1}=-\zeta w_{2}$. 
Notice that we are done if either
of the sets $J_1$ or $J_2$ is empty.


What we have done so far, can be considered as a high-level version of 
the approach of our earlier work~\cite{ISS} with Sanselme. 
The method of \cite{ISS} then proceeds with 
recursion to $m-1$. Unfortunately, that approach is appropriate
only in the quadratic case. Here we use a completely different idea which 
will turn to be extensible to more general degrees.

{From} the vectors in $V$ we form $m+1$ pairwise disjoint sets of vectors 
of size $m+1$. By the construction above, we compute 
$w_{1}(1)$, $w_{2}(1),\ldots$, $w_{1}(m+1)$, $w_2(m+1)$, where 
\begin{equation}
w_{1}(i)=-\zeta w_{2}(i),
\label{dependence}
\end{equation}
for $i = 1, \ldots, m+1$. Moreover,
these
$2m$ vectors are represented as linear combinations with nonzero 
square coefficients of $2m$ pairwise disjoint nonempty subsets of the 
original vectors.

Now $w_1(1),\ldots,w_1(m+1)$ are linearly dependent and
again we can find disjoint subsets $J_{1}$ and $J_{2}$ and scalars
$\gamma_i$ for $i\in J_{1}\cup J_{2}$ such that for
$w_{11}=\sum_{i\in J_{1}}\gamma_i^2w_{1}(i)$ and
$w_{12}=\sum_{i\in J_{2}}\gamma_i^2w_{1}(i)$ we have
$w_{11}=-\zeta w_{12}$. But then for
$w_{21}=\sum_{i\in J_1}\gamma_i^2w_{2}(i)$ and
$w_{22}=\sum_{i\in J_2}\gamma_i^2w_{2}(i)$, using equation~(\ref{dependence}) for all $i$,
we similarly have $w_{21}=-\zeta w_{22}$.
On the other hand, if we sum up equation~(\ref{dependence}) for $i \in J_{1}$, we get
$w_{11}=-\zeta w_{21}$. Therefore
$$w_{11} =  \zeta^2 w_{22} ~~ \mbox{{\rm and }}~~ w_{12} = w_{21} = - \zeta w_{22}.$$
By Fact~\ref{fact:woestijne} we can find field elements 
$\delta_{11},\delta_{22},\delta_{12}$,
not all zero, such that $\zeta^2 \delta_{11}^2 - 2\zeta\delta_{12}^2  + \delta_{22}^2 =0,$
and therefore $(\zeta^2 \delta_{11}^2 - 2\zeta\delta_{12}^2  + \delta_{22}^2) w_{22} = 0.$
But $$(\zeta^2\delta_{11}^2 - 2\zeta\delta_{12}^2  +\delta_{22}^2) w_{22} =
\delta_{11}^2w_{11}+
\delta_{12}^2(w_{12}+w_{21})+\delta_{22}^2w_{22}.$$
Then expanding 
$\delta_{11}^2w_{11}+
\delta_{12}^2(w_{12}+w_{21})+\delta_{22}^2w_{22}=0$
gives a representation of the zero vector as a linear combination
with square coefficients (squares of appropriate product of 
\mbox{$\beta$}s, \mbox{$\gamma$}s and \mbox{$\delta$}s) of a subset of the original vectors.
\end{proof}

\subsection{The cubic case}
\begin{proposition}\label{prop:cubic}
Let $n= (9m+1)(3m+1)(m+1)$. 
Then 
$\SDE(\F_q, n, m, 3)$ can be solved by a randomized
algorithm in time polynomial in $m$ and $\log q$.
\end{proposition}
\begin{proof}
We assume that $q-1$ is divisible by $3$  since otherwise the problem is trivial.
By a randomized polynomial time algorithm we can compute 
two elements $\zeta_2,\zeta_3~$from $\F_q$ such that
$\zeta_1=1,\zeta_2,\zeta_3$ are a complete set of representatives
of the cosets of the subgroup $\{x^3 : x \in \F_q^*\}$ of 
$\F_q^*$. Let $V$ be our input set of $n$ vectors in $\F_q^m$, now we want to represent the
zero vector as a nontrivial linear combination of some vectors from $V$ where all the coefficients are cubes.

As in the quadratic case, for any subset of
$m+1$ vectors $v_1,\ldots,v_{m+1}$ from $V$, we can easily find a proper linear combination
summing to zero,
$\sum_{i=1}^{m+1}\alpha_iv_i = 0$.
For $r=1,2,3,$ let $J_r$ be the set of indices such that
$0 \neq \alpha_i=\beta_i^3\zeta_r$. We know that at least one of these three sets is non-empty.
 For each $\alpha_i\neq 0$ we can efficiently
identify the coset of $\alpha_i$ and even find $\beta_i$ using
the method of \cite{Adleman}. 
Let $w_r=\sum_{i\in J_r}\beta_i^3v_i$. Then $\zeta_1w_1+\zeta_2 w_2 +\zeta_3
w_3=0$. Without loss of generality we can suppose that $J_1$ is non-empty
since if $J_r$ is non-empty for $r \in \{2,3\}$, we can just multiply
the $\alpha_i$s simultaneously by $\zeta_1/\zeta_r$.

From any subset of size $(3m+1)(m+1)$ of $V$ 
we can form $3m+1$ groups
of size $m+1$, and within each group we can do the procedure outlined above.
This way we obtain, for  $k=1,\ldots,3m+1$, and $r=1,2,3$,
pairwise disjoint subsets
$J_r(k)$ of indices and vectors $w_r(k)$ such that
\begin{equation}
\zeta_1w_1(k)+\zeta_2 w_2(k) +\zeta_3w_3(k)=0.
\label{cubic_dependence}
\end{equation}
For $k=1,\ldots,3m+1$, we know that
$J_1(k)\not=\emptyset$ and the vectors $w_r(k)$ are combinations
of input vectors with indices form $J_r(k)$ having coefficients
which are nonzero cubes.
Let $W(k) \in \F_q^{3m}$ denote the vector obtained by concatenating
$w_1(k)$, $w_2(k)$ and $w_3(k)$ (in this order). Then
we can find three pairwise disjoint subsets $M_1,M_2,M_3$
of $\{1,\ldots,3m+1\}$, and for each $k\in M_s$, a nonzero field
element $\gamma_k$ such that 
\begin{equation}
\sum_{s=1}^3\zeta_s\sum_{k\in M_s}\gamma_k^3W(k)=0.
\label{cubic_second_dependence}
\end{equation}
We can arrange that $M_2$ is non-empty.
For $r,s\in\{1,2,3\}$, set $J_{rs}=\bigcup_{k\in M_s}J_r(k)$
and $w_{rs}=\sum_{k\in M_s}\gamma_k^3w_r(k)$.
Then $w_{rs}$ is a linear combination of input vectors 
with indices from $J_{rs}$ having coefficients that
are nonzero cubes. The equality~(\ref{cubic_second_dependence})
just states that $\zeta_1w_{r1}+\zeta_2w_{r2}+\zeta_3w_{r3}=0$, for $r=1,2,3$.
Furthermore, summing up the equalities~(\ref{cubic_dependence}) 
for $k \in M_s$, we get $\zeta_1w_{1s}+\zeta_2w_{2s}+\zeta_3w_{3s}=0$,  for $s=1,2,3$.

Continuing this way, from $(9m+1)(3m+1)(m+1)$ 
input vectors we can make 27 linear combinations with
cubic coefficients $w_{rst}$, for $r,s,t =1,2,3$,
having pairwise disjoint supports such that the support
of $w_{123}$ is non-empty and they satisfy the 27
equations 
\begin{center}
$\zeta_1w_{1st}+\zeta_2w_{2st}+\zeta_3w_{3st}=0$
($s,t=1,2,3$);
\\
$\zeta_1w_{r1t}+\zeta_2w_{r2t}+\zeta_3w_{r3t}=0$
($r,t=1,2,3$);
\\
$\zeta_1w_{rs1}+\zeta_2w_{rs2}+\zeta_3w_{rs3}=0$
($r,s=1,2,3$). 
\end{center}
{From} these
we use the following $6$ equations:
\begin{center}
$\zeta_1w_{123}+\zeta_2w_{223}+\zeta_3w_{323}=0$;
\\
$\zeta_1w_{132}+\zeta_2w_{232}+\zeta_3w_{332}=0$;
\\
$\zeta_1w_{213}+\zeta_2w_{223}+\zeta_3w_{233}=0$;
\\
$\zeta_1w_{312}+\zeta_2w_{322}+\zeta_3w_{332}=0$;
\\
$\zeta_1w_{231}+\zeta_2w_{232}+\zeta_3w_{233}=0$;
\\
$\zeta_1w_{321}+\zeta_2w_{322}+\zeta_3w_{323}=0$.
\end{center}
Adding these equalities with appropriate signs so that
the terms with coefficients $\zeta_2$ and $\zeta_3$
cancel and dividing by $\zeta_1$, we obtain
\begin{equation}
\label{eqn:sixterm}
w_{123}+w_{231}+w_{312}-w_{132}-w_{213}-w_{321}=0.
\end{equation}
Observing that $-1 = (-1)^3$, this gives 
a representation of zero
as a linear combination of the input vectors
with coefficients that are cubes. 
(Note that the algorithm described in this proof does not
rely on van~de~Woestijne's result Fact~\ref{fact:woestijne}. This
is because we were in a position to eliminate the $\zeta_i$s and obtained
a linear dependency with coefficients $\pm 1$ which are always cubes of
themselves in $\F_q$, independently of $q$.)


\end{proof}

\section{The general case}
In this section we prove Theorem~\ref{thm:sde}. First we make the simple observation that it is sufficient
to solve $\SDE(\F_q,n,m,d)$ in the case when $d$ divides $q-1$. If it is not the case, then
let $d'=\gcd(d, q-1)$. Then from a nonzero solution of the system 
$$\sum_{j=1}^n x_j^{d'} v_j =0,$$
one can efficiently find a nonzero solution of the original
equation.
Indeed, 
the extended Euclidean algorithm efficiently
finds a positive integer $t$ such that $td = u(q-1) + d'$ for some integer $u$.
Then for any nonzero $x\in \F_q$ we have $(x^t)^d = x^{d'} \mod p$,
and therefore $(x_1^t, \ldots, x_n^t)$ is a solution of equation~(\ref{eq:vector}).

From now on we suppose that $d$ divides $q-1$.
Our algorithm will consist of two major procedures. The first one
is devoted to finding two disjoint subsets of the input vectors, not both
empty, and $d$th power coefficients such that the linear combinations
of the vectors from the two subsets give equal vectors. Notice that
this part already does the job when one of the two sets happen to
be empty or $d$ is odd (or, more generally, a $d$th root of $-1$
is at hand). The second procedure consists of iterative applications
of the first algorithm to obtain a vector with sufficiently many
representations as linear combinations with $d$th power coefficients 
with pairwise disjoint supports.

We will denote by $C(d,m)$ the number of vectors (variables) used by 
our algorithm. For $d=1$, we can obviously take $C(1,m) = m+1$.

The basic idea of the first algorithm is -- like in the cubic and
quadratic case outlined in the previous section -- getting
linear dependencies and effectively putting the coefficients
of these dependencies into cosets of the multiplicative group
of the $d$th powers on nonzero field elements. In the first
subsection, based on an idea borrowed from \cite{vdW},
 we show how to do this without having nonresidues
at hand.

\subsection{Classifying field elements}
\label{subsec:class}

During the procedures of this section, one of the basic tasks is
the following. Given a nonzero field element $\alpha$, one has to
write $\alpha$ as $\alpha=\zeta_i\beta^d$, where $1=\zeta_1$,
$\ldots$, $\zeta_d$ are fixed elements. Ideally, the $\zeta_i$
form a complete system of representatives of the cosets of
the subgroup of the $d$th powers in the multiplicative group
$\F_q^*$. Unfortunately, no deterministic polynomial time
algorithm is known 
to find an element of a nontrivial coset
(unless assuming the generalized Riemann hypothesis). Therefore,
instead of the whole $\F_q^*$, we consider (roughly speaking)
the subgroup generated by nonzero field elements already 
seen and we classify elements according to the cosets
of $d$th powers of this subgroup. The classification fails
(essentially) when we encounter an element outside this group.
Then the subgroup, the sub-subgroup of its $d$th powers as well as
the coset representatives are updated and all the computations
done so far are redone. Obviously, this can happen at most
$\log q$ times, resulting a $\log q$ factor in complexity
(but not in the bound on the number of input vectors necessary for
success).

To describe the details, we need some notation. Let $\pi$
be the set of prime divisors of $d$ and $\pi'$ be the set
of prime divisors of $q-1$ outside $\pi$. Then the multiplicative
group $\F_q^*$ is the (direct) product of two subgroups $H_\pi$ and
$H_{\pi'}$,
where $H_\pi$ consists of the elements of order having prime factors
from $\pi$, while the element of $H_{\pi'}$ are those having an order
whose prime factors are from $\pi'$. Note that the primes in $\pi$
can be computed in time $d^{O(1)}$ by factoring $d$. The primes
in $\pi'$ do not need to be explicitly computed. Instead, by successively
dividing $q-1$ by the primes in $\pi$, we can efficiently (that is, in
time polynomial in $\log q$) compute
the order of the subgroup $H_\pi$, which is the largest divisor
of $q-1$ coprime to $d$. Given an element $\alpha\in \F_q^*$,
one can find in time polynomial in $\log q$ the unique elements 
$\gamma\in H_\pi$ and $\gamma'\in H_{\pi'}$ such that $\alpha=\gamma\gamma'$ (see,
e.g., \cite{vdW} for details). Also, one can efficiently find
the unique element $\delta'\in H_{\pi'}$ such that $\gamma'={\delta'}^d$.
(Actually, $\delta'={{\gamma'}^r}$ where $rd\equiv 1$ modulo the order 
of $H_{\pi'}$.) 

Instead of $H_\pi$ we use the subgroup $H$ of the $\pi$-parts
of the field elements given so far to the classification
procedure as input. We assume that $H$ is given by a generator
$\eta$. Elements $1=\zeta_1,\ldots,\zeta_d\in H$ are also assumed
to be given such that they form a possibly redundant, but complete
system of representatives of cosets of the subgroup $H^d$ consisting
of the $d$th powers from $H$. Initially $\eta=1=\zeta_1=\ldots=\zeta_d$.
Given $\alpha=\gamma\gamma'$, we (attempt to) compute the $\eta$-base 
discrete logarithm of $\gamma$ using the method of Pohlig and Hellman~\cite{Pohlig}. This takes time polynomial in $d$ and $\log q$.
In the case of success, we can use the logarithm to locate the coset
of $\gamma$ and write $\gamma$ as $\gamma=\delta^d\zeta_i$ where $\delta\in
H$. Then $\alpha=\beta^d\zeta_i$, where $\beta=\delta\delta'$.

In the case of failure, we replace $\eta$ by a generator
of the subgroup generated by $\gamma$ and $\eta$ and 
we replace $\zeta_2,\ldots,\zeta_d$ by $\eta$,$\ldots$,$\eta^{d-1}$
(repetitions may occur). We restart the whole algorithm with
these new data.

\subsection{Finding colliding representations}

\label{subsec:twofold}

In this subsection we prove the following.

\begin{theorem}
\label{thm:twofold}
Assume that $d|q-1$ and put $G(d,m)=d^{\frac{d(d-1)}{2}}(m+1)^d$. 
Then, given $G=G(d,m)$ input vectors $v_1,\ldots,v_G\in \F_q^m$,
in time polynomial in $G$ and $\log q$, we can find two disjoint 
subsets $I$ and $J$ of $\{1,\ldots,G\}$ with $I\neq \emptyset$ and 
nonzero field elements $\gamma_j\in \F_q^*$ ($j\in I \cup J$) 
such that $\sum_{i\in I}\gamma_i^dv_i=\sum_{j\in J}\gamma_j^dv_j$.
\end{theorem}

\begin{proof}
The algorithm follows the lines already presented in the proof
of Proposition~\ref{prop:cubic} for the cubic case. The main difference is
that here we (possibly) need more rounds of iteration.
For $\ell=1,\ldots,d$,
put $B_\ell(d,m)=d^{\frac{\ell(\ell-1)}{2}}(m+1)^\ell$. 
For ${\underline a} = (a_1, \ldots , a_{\ell})\in \{1,\ldots,d\}^\ell$, for $s \in \{1, \dots, d\}$ and 
for $1 \leq j \leq \ell$, set
$${\underline a}(j,s) =(a_1,\ldots,a_{j-1},s,a_{j+1},\ldots,a_{\ell}).$$

\begin{lemma}
\label{lemma:dcube}
From $B=B_\ell(d,m)$ input vectors 
$v_1,\ldots,v_B$, in time polynomial
in $B$ and $\log q$, we can 
find $d^\ell$ pairwise disjoint subsets 
$J_{\underline a}\subseteq \{1,\ldots,B\}$ and field elements $\beta_1,\ldots,\beta_B$ such that
$J_{(1,\ldots,\ell)}\neq \emptyset$, and if we set 
$w_{\underline a}=\sum_{i\in J_{\underline a}}\beta_i^dv_i,$ then 
we have 
$$\sum_{s=1}^{d}\zeta_sw_{{\underline a}(j,s)}=0,$$
for every ${\underline a}\in \{1,\ldots,d\}^\ell$ and
$j=1,\ldots,\ell$.
\end{lemma}

\begin{proof}
We prove it by recursion on $\ell$.
If $\ell=1$ then any $B_\ell(d,m)=m+1$ vectors from $\F_q^m$
are linearly dependent. Therefore there exist
$\alpha_1,\ldots,\alpha_{m+1}\in\F_q$, not all zero, 
such that $\sum_{i=1}^{m+1}\alpha_iv_i=0$. Using the procedure
of Subsection~\ref{subsec:class}, we find subsets $J_1,\ldots,J_d$
of $\{1,\ldots,m+1\}$ and field elements $\beta_i$ 
($i\in J_1\cup\cdots\cup J_d$), such that for $i\in J_r$
we have $\alpha_i=\zeta_r\beta_i^d$. At least
one of the sets $J_r$ is non-empty. If $J_1$ is empty then we
multiply  the coefficients $\alpha_i$ simultaneously by $\zeta_1/\zeta_r^{-1}$ where
$J_r$ is nonempty to arrange that $J_1$ becomes nonempty.

To describe the recursive step, assume that we are given 
$B_{\ell+1}(d,m)= d^\ell(m+1) B$ vectors. Put
$E=d^{\ell}(m+1)$, and for convenience 
assume that the input vectors are denoted by $v_{ki}$, for
$k=1,\ldots,E$ and  $i=1,\ldots,B$. By the recursive hypothesis, 
for every $k\in\{1,\ldots,E\}$, there exist subsets
$J_{\underline a}(k)\subseteq \{1,\ldots,B\}$
and field elements $\beta_i(k)$ such that
$J_{(1,\ldots,\ell)}(k) \neq \emptyset$, and
with $w_{\underline a}(k)=\sum_{i\in J_{\underline a}(k)}\beta_i(k)^dv_{ki}$,
we have 
\begin{equation}
\sum_{s=1}^{d}\zeta_sw_{{\underline a}(j,s)}(k)=0,
\label{d_dependence}
\end{equation}
for every ${\underline a}\in \{1,\ldots,d\}^\ell$ 
and $j=1,\ldots,\ell$. 

For every $k=1,\ldots,E$, let $W(k)$ be the concatenation of the vectors
$w_{\underline a}(k)$ in a fixed, say the lexicographic,
order of $\{1,\ldots,d\}^\ell$. Then the $W(k)$'s are vectors of length
$d^\ell m<E$. Therefore there exist field elements
$\alpha(1),\ldots,\alpha(E)$, not all zero, such that
$\sum_{k=1}^E \alpha(k) W(k)=0$. For a $k$ such that
$\alpha(k)\neq 0,$ let $\alpha(k)=\zeta_r\gamma(k)^d$ for some $1 \leq r \leq d$
and $\gamma(k)\in \F_q^*$.
The index $r$
and $\gamma(k)$ are computed by the procedure of
Subsection~\ref{subsec:class}. 
For $r=1,\ldots,d$,
let $M_r$ be the set of $k$'s such that $\alpha(k)=\zeta_r\gamma(k)^d$.
We can arrange that $M_{\ell+1}$ is non-empty by simultaneously 
multiplying the $\alpha(k)$'s
by $\zeta_{\ell+1}/\zeta_r$ for some $r$, if necessary. Observe that we have
\begin{equation}
\sum_{s=1}^{d}\zeta_s\sum_{k\in M_s}\gamma(k)^dW(k)=0.
\label{d_second_dependence}
\end{equation}

For $i\in \{1,\ldots,B\}$ and
$k\in \{1,\ldots,E\}$ set $\beta'_{ki}=\gamma(k)\beta_i(k)$.
We fix
${\underline a}'\in \{1,\ldots,d\}^{\ell+1}$,
and we set $\underline a = (a_1', \ldots a_{\ell}')$ and $r = a_{\ell + 1}'$.
We define
$J'_{{\underline a}'} = \{(k,i):k\in M_r\mbox{~and~}i\in 
J_{\underline a}(k)\}$ and
$w'_{{\underline a}'} =
\sum_{(k,i)\in J'_{{\underline a}'}}{\beta'}_{ki}^dv_{ki}$.
Then 
$w'_{{\underline a}'}=\sum_{k\in M_r}\gamma_k^dw_{\underline a}(k)$.
This equality, together with the equalities~(\ref{d_dependence})
imply that for every $j=1,\ldots,\ell$, we have
$$\sum_{s=1}^{d}\zeta_sw'_{{\underline a}'(j,s)}=0.$$
For $j=\ell+1$ consider the equality~(\ref{d_second_dependence}),
from which follows that 
$$\sum_{s=1}^{d}\zeta_s\sum_{k\in M_s}\gamma(k)^dw_{\underline a}(k)=0.$$
Expanding 
$w_{\underline a}(k)$ in the inner sum
$\sum_{k\in M_s}\gamma(k)^dw_{\underline a}(k)$ gives that
it equals $w'_{{\underline a}'(\ell +1,s)}.$
Thus also
$$\sum_{s=1}^{d}\zeta_sw'_{{\underline a}'(\ell + 1,s)}=0,$$
finishing the proof of the lemma.
\end{proof}

\medskip
We apply the procedure of Lemma~\ref{lemma:dcube} for $\ell=d$. From
$B=B_d(d,m)=d^{\frac{d(d-1)}{2}}(m+1)^d$ input vectors
$v_1,\ldots,v_B$, we compute in time polynomial in
$\log q$ and $B$ subsets $J_{\underline a}$,
with
$J_{(12\ldots d)} \neq \emptyset$, as well as nonzero
elements $\beta_1,\ldots,\beta_B\in \F_q$ such that
with
$w_{\underline a}=\sum_{i\in J_{\underline a}}\beta_i^dv_i,$
we have
\begin{equation}
\sum_{s=1}^d\zeta_s w_{{\underline a}(j,s)}=0,
\label{whatever}
\end{equation}
for every $j=1,\ldots,d$ and for every ${\underline a}\in \{1,\ldots,d\}^d$.

Tuples from $\{1,\ldots,d\}^d$ without repetitions are of special interest.
We identify such a $d$-tuple ${\underline a}=(a_1,\ldots,a_d)$ with 
the permutation $i\rightarrow a_i$ from the symmetric group $S_d$ on 
$\{1,\ldots,d\}$. With some abuse of notation, we denote this permutation
also by $\underline a$.
By $\mbox{sgn}({\underline a})$ we denote the {\em sign} of $\underline a$,
considered as a permutation. The sign of $\underline a$ is $1$ if ${\underline a}$ 
is even and $-1$ if ${\underline a}$ is odd.
We show that  
\begin{equation}
\sum_{{\underline a}\in S_d}\mbox{sgn}({\underline a})w_{\underline a}=0.
\label{66}
\end{equation}
For ${\underline a}\in S_d$, let $j_{\underline a}$ be the position 
of $1$ in ${\underline a}$ and for every $s\in \{1,\ldots,d\}$, we denote
by ${\underline a}[s]$ the sequence obtained from $a$ by replacing $1$ with
$s$. Notice that ${\underline a}[s]={\underline a}(j_{\underline a},s)$,
therefore~(\ref{whatever}) implies
\begin{equation}
\label{whatever1}
\sum_{{\underline a}\in S_d}\mbox{sgn}({\underline a})
\sum_{s=1}^d\zeta_s w_{{\underline a}[s]}=0.
\end{equation}
We claim that
\begin{equation}
\label{whatever2}
\sum_{\underline a\in S_d}
\mbox{sgn}({\underline a})\sum_{s=2}^d\zeta_s w_{{\underline a}[s]}=0.
\end{equation}
To see this, observe that for $s>1$ the tuple
${\underline a}[s]$ has entries from $\{2,\ldots,d\}$, where
$s$ occurs twice, while the others once. Any such sequence ${\underline a}'$
can come from exactly two permutations which differ by a transposition:
these are obtained from ${\underline a}'$ by replacing one of the
occurrences of $s$ with $1$. Then~(\ref{66}) is just the difference of 
equalities~(\ref{whatever1}) and~(\ref{whatever2}).

Put 
$$I=\bigcup_{{\underline a}\mbox{\scriptsize~even}}J_{\underline a},\;\;\;
 J=\bigcup_{{\underline a}\mbox{\scriptsize~odd}}J_{\underline a}\mbox{~~~and~~~}
\gamma_i=\beta_i \mbox{~~~for~} i\in I \cup J.$$
 (Here, 
${\underline a}\mbox{~even}$ resp.~${\underline a}\mbox{~odd}$
abbreviates that ${\underline a}$ is an even or an odd permutation,
respectively.)
Then~(\ref{66}) gives the desired pair of colliding representations.
\end{proof}

\subsection{Accumulating collisions}

In this subsection we finish the proof of 
Theorem~\ref{thm:sde}.

\begin{proof}[Proof of Theorem~\ref{thm:sde}]
We assume that $q-1$ is divisible by $d$. By Theorem~\ref{thm:twofold},
from $G(d,m)$ input vectors we can select two disjoint subsets,
not both empty, and find $d$th power coefficients such that 
the corresponding linear combinations represent the same vector.
Notice that we are done if this is the zero vector.

When we have $G(d,m)^2$ input vectors, the procedure
of Theorem~\ref{thm:twofold}, applied to $G(d,m)$ groups of size $G(d,m)$,
gives $G(d,m)$ vectors and two representations as linear
combination with $d$th power coefficients for each. (These
combinations have $2G(d,m)$ pairwise disjoint sets as support.)
Applying the procedure again to the $G(d,m)$ vectors and multiplying
the coefficients gives a vector with 4 representations as linear
combinations having pairwise disjoint support and $d$th power
coefficients. 

Iterating this, using $G(d,m)^{\ell}$ input vectors, we obtain
a vector with $2^\ell$ representations 
as linear
combinations having pairwise disjoint support and 
coefficients that are explicit $d$th powers. When $2^\ell\geq d+1$, we can use 
Fact~\ref{fact:woestijne} to find field elements $z_1,\ldots, z_{d+1}$,
not all zero, such that $z_1^d+\ldots+z_{d+1}^d=0$. 
Multiplying the coefficients of the $i$th representation by $z_i^d$ we obtain the desired representation 
of the zero vector. We have
$$C(d,m)\leq G(d,m)^{\lceil \log_2(d+1) \rceil}\leq
d^{d^2\log d}(m+1)^{d\log d}.$$
\end{proof}

\section{Application in Quantum computing}

\label{sec:qappl}

\subsection{Reduction from the special \HSP{} to \HPGP'{}}

In this part we give the details of a reduction from
a special instance of the hidden subgroup problem
in groups which are semidirect products of an elementary
abelian $p$-groups by a group of order $p$. The arguments
here are quite standard. 

\begin{proof}[Proof of Proposition~\ref{prop:hsp_to_hpgp}]
A semidirect product group of the form $\F_p\ltimes \F_p^m$ can be specified by an
automorphism of $\F_p^m$. The automorphisms of $\F_p^m$ can be identified with nonsingular
$m\times m$ matrices $B$ over $\F_p$ such that $B^p=I$. For such a matrix $B$, the group
$G_B = \F_p \, \leftidx{_B}{\ltimes}{} \F_p^m$
can be represented
as the set of $(m+1)\times (m+1)$ matrices over $\F_p$ 
$$
\left \{\begin{pmatrix}
B^x & v \\
0 & 1
\end{pmatrix} ~~:~~ x \in \F_p, ~ v \in \F_p^m \right \}.
$$
We choose the
{quantum} encoding $\ket{x}\ket{v}$ for the matrix
$$M_B(x,v)=\begin{pmatrix}
B^x & v \\
0 & 1
\end{pmatrix}.$$
Let 
$$
K = 
\left \{ \begin{pmatrix}
B^x & 0 \\
0 & 1
\end{pmatrix} ~~:~~ x \in \F_p \right \}
\mbox{~~~and~~~}
N = 
\left \{ \begin{pmatrix}
I & v \\
0 & 1
\end{pmatrix}~~:~~ v\in \F_p^m \right \}.
$$
Then
$N$ is a normal subgroup of $G$ of index $p$
and $K\cap N=\{1_G\}$. 
For every $v \in \F_p^m$, consider the {cyclic subgroup}
$$H_v = \left \langle  \begin{pmatrix}
B & v\\
0 & 1
\end{pmatrix}  \right \rangle = 
\left \{\begin{pmatrix}
B^x & v(x) \\
0 & 1
\end{pmatrix} ~~:~~ x \in \F_p \right \},
$$
where
$$
v(x) = 
 \begin{pmatrix}
  v_{1}(x)  \\
  \vdots  \\
  v_{m}(x)
 \end{pmatrix}
=
(B^{x-1} + \cdots + B^1 + B^0 ) v.
$$
Then ${\cal H}$, the family of subgroups of $G_B$ of order $p$ which are not subgroups of $N$ is exactly
$ \{H_v : v \in \F_p^m\}$. The hidden function hides some member of ${\cal H}$.
Since $B^p = I $ we also have $ (B-I)^p = 0$. It can be seen that if the nilpotency class of
$G_B$ is $d$ then $d$ is the smallest integer such that $(B - I)^d =0$.
\suppress{
or generated by a matrix of the form
$$
D=\begin{pmatrix}
B & * \\
0 & 1
\end{pmatrix}.
$$
Notice that $D^p=I$ therefore $(D-I)^p=0$. 
Put $E=\sum_{j=1}^{d}\frac{{-1}^{j-1}}{j}(D-I)^j..$
(This is $\log (D-I)$.) Then $E^p=0$ and
$$D^t=\exp(tE)=\sum_{j=0}^{d}\frac{1}{j!}(tE)^j.$$
I fact, $E$ is of the form 
$$
E=\begin{pmatrix}
A & v \\
0 & 0
\end{pmatrix}.
$$
where $A^p=0$ and 
$$B^t=\exp(tA)=\sum_{j=0}^{d}\frac{1}{j!}t^j(A)^j.$$
Let $d$ be the smallest positive
integer $d$ such that $A^d=0$. It can be seen that $d$
is the nilpotency class of $G$. 

In fact, the lower central
series of $G$ is the sequence consisting of the images
of $A$, $A^2$, etc.
}
In fact, if we let $A = \log B$ then the lower central
series of $G_B$ is the sequence consisting of the images
of $A, A^2, \ldots, A^{d-1}$.

\begin{claim}
The functions $v_i(x)$ are polynomials with $0$ constant term and of degree $ \leq d$, for
$i = 1, \ldots, m$.
\end{claim}
\begin{proof}
We have
$$
A = \log B = \sum_{j=1}^{d-1}\frac{{-1}^{j-1}}{j}(B-I)^j.
$$
Then
$$
B^k= e ^{kA} = \sum_{j=0}^{d-1}\frac{A^j}{j!}k^j,$$
since $A^d = 0$. Therefore
\begin{align*}
v(x) & = \sum_{k=0}^{x-1} B^k v \\
& = \sum_{k=0}^{x-1} \sum_{j=0}^{d-1}\frac{A^jv}{j!}k^j \\
& = \sum_{j=0}^{d-1} \frac{A^jv}{j!} \sum_{k=0}^{x-1} k^j\\
& = \sum_{j=0}^{d-1} \frac{A^jv}{j!} p_j(x-1),
\end{align*}
where $p_0(x-1) = x$, and $p_j(x)$ is a degree $j+1$ polynomial expressed by the Faulhaber's formula,
for $j = 1,\ldots, d-1$.
It is known~\cite{jacobi}  that $p_j(x)$ is divisible by $x+1$, for all $j$. Therefore indeed $v_i(x)$ is a degree 
$\leq d$ polynomial with
constant member zero, for $i=1, \ldots, m$.
\end{proof}
Let us now suppose that our input $f$ to $\HSP(G_B, {\cal H})$
hides the subgroup 
$$H_v =  
\left \{\begin{pmatrix}
B^x & v(x) \\
0 & 1
\end{pmatrix} ~~:~~ x \in \F_p \right \}.
$$
We can take as coset representatives
$$
N = \left \{\begin{pmatrix}
I & u \\
0 & 1
\end{pmatrix} ~~:~~  u \in \F_p^m \right \}.
$$
Since
$$
\begin{pmatrix}
I & u \\
0 & 1
\end{pmatrix}
\begin{pmatrix}
B^x & v(x) \\
0 & 1
\end{pmatrix} =
\begin{pmatrix}
B^x & u + v(x) \\
0 & 1
\end{pmatrix},
$$
the left cosets of $H_v$ are of the form
$$
 \left\{ \begin{pmatrix}
B^x & u+v(x) \\
0 & 1
\end{pmatrix}
~~:~~ x \in \F_p \right \}=
\left\{M_B(x,u+v(x))
~~:~~ x \in \F_p \right \}
,$$
for $ u\in \F_p^m$. By a standard efficient quantum procedure 
we can create, for a random $u \in \F_p^m$, the {coset state}
$$\sum_{x \in \F_p}\ket{x}\ket{u+v(x)}.
$$
But this is also  a random { level set state} of the function 
$$f : \F_p \times \F_p^m \rightarrow \F_p^m, ~~~~~ f(x,y) = y - v(x),$$
and therefore the input to $\HPGP'(\F_p, 1,m,d)$ hiding 
the polynomial $v(x)$. From the solution $v(x)$ we can recreate the solution of the $\HSP$ problem
since $v = v(1)$.

\end{proof}

\section{Proof of Theorem~\ref{thm:hpgp_to_sde}}

In this part we outline a modified version of the method
of our work~\cite{DHIS} with Decker and H{\o}yer. A critical ingredient is solving 
systems of diagonal polynomial equations with sufficiently
many variables. At the time of writing \cite{DHIS} polynomial time
algorithms (except for the cases $d=1,2$) were available only
for the case when the number of equations is constant.)
Now we have a version which works in polynomial time
even if $m$ is not constant.

\begin{proof}[Proof of Theorem~\ref{thm:hpgp_to_sde} (sketch)]
%
A solution for constant $p$ is given in \cite{DISW}. 
(Interestingly, that solution goes through a reduction 
to the variant of the hidden subgroup problem with coset
states as input in a $p$-group of nilpotency 
class $d+1$ and exponent $p$. The latter problem is solved by
the method of the paper
\cite{FIMSS03} by the authors with Friedl, Magniez and Shen,
which works efficiently in groups of
constant derived length and constant exponent.) Thefore 
we may assume that $p>d$. Although this assumption is 
not essential, it simplifies presentation very much.

The input for \HPGP'{} consists of uniform superpositions
of random level sets states of the form (\ref{eq:levstate}),
which, for the special case we have are states
$$\ket{x}\ket{u+\sum_{j=1}^dx^jw_j},$$
for random (unknown) $u \in \F_p^m$.
To handle dependency on $u$, we apply the Fourier transform 
of $\F_p^m$ to the second register of such a state.
The result is
$$\omega^{\sum_{k=1}^my_ku_k}\sum_{x=0}^{p-1}
\omega^{\sum_{j=1}^d x^j\sum_{k=1}^my_kw_{jk}}\ket{x}\ket{y}
=\omega^{\sum_{k=1}^my_ku_k}\ket{\phi_y}\ket{y},$$
where $\omega=\sqrt[p]{1}$ and
$$\ket{\phi_y}=
\sum_{x=0}^{p-1}\omega^{\sum_{j=1}^d x^j\sum_{k=1}^my_kw_{jk}}\ket{x}.$$
Measuring the second register we obtain, up to a global phase,
the state $\ket{\phi_y}$ with known $y$. We drop the useless
states $\ket{\phi_0}$. It can be seen that
each $y\in \F_p^m$ occurs with equal probability,
therefore $\ket{\phi_0}$ occurs with probability $\frac{1}{p^m}$.

We rewrite $\ket{\phi_y}$ in a more general form suitable for recursion.
For hidden parameters $\hparam_1,\ldots,\hparam_\ell\in \F_p$ and for $Y\in \F_p^{d\times \ell}$
let $$\ket{\psi_Y}:=\sum_{x=0}^{p-1}
\omega^{\sum_{j=1}^d x^j\sum_{k=1}^\ell Y_{jk}\hparam_k}\ket{x}.$$
In words, the coefficient of $x^j$ in the phase of the state $\ket{\psi_Y}$
is a linear combination of the hidden parameters  
with known coefficients $Y_{j1},\ldots,Y_{j\ell}$.
 Then
$\ket{\phi_y}=\ket{\psi_Y}$, where $\ell=dm$,
$\hparam_{(j-1)d+k}=w_{jk}$, $Y_{j,(j-1)d+k}=y_{k}$,
and $Y_{j,(j'-1)d+k}=0$, for $j,j'=1,\ldots,d$, $j'\neq j$, 
$k=1,\ldots,m$. The goal is to determine the hidden
parameters $\hparam_1,\ldots,\hparam_\ell$.

Let $n=n(\ell,d)$ be a positive integer such that
for any positive integer $d'\leq d$ nonzero solutions 
of systems of equations of the form
$$\sum_{j=1}^n a_{ij} \xi_j^{d'}=0, \mbox{~~~~{\rm for} }i=1,\ldots,\ell,$$
in the variables $\xi_1,\ldots,\xi_n$ can be found in time
polynomial in $n\ell\log p$. 

Using $n$ level set superpositions, we obtain 
$n$ states of the form $\ket{\psi_Y}$ with various $Y$.
More precisely, up to a global phase we obtain a state
$$
\ket{\psi_{Y^1}}\ldots\ket{\psi_{Y^n}}=
\sum_{x_1,\ldots,x_n=0}^{p-1}
\omega^{
\sum_{j=1}^{d} (x_1^j \sum_{k=1}^\ell Y^1_{jk}\hparam_k+\ldots+
x_n^j\sum_{k=1}^\ell Y^n_{jk}\hparam_k )
}\ket{x_1,\ldots,x_n}.$$
If the degree $d$ term is completely missing from
the phase of state $\ket{\psi_{Y^i}}$, that is, $Y^i_{dk}=0$ for 
$k=1,\ldots,\ell$,
then we take $\ket{\psi_{Y^i}}$ and ignore all the other states. Otherwise
we produce a similar state without degree $d$ term as follows.
(This is the point where the new algorithm differs from
that of our eralier work \cite{DHIS} with Decker and H{\o}yer. Originally the degree $d$ terms
had to be eliminated one-by-one which caused an exponential blowup
of the costs in $m$. The main result of the present paper allows 
us to eliminate all the degree $d$ terms simultaneously, in one step,
saving the exponential blowup.)

We find a nonzero solution $(\delta_1,\ldots,\delta_n)\in \F_p^n$
of the system of equations $\sum_{i=1}^n\delta_i^dY^i_{k}=0$, for
$k=1,\ldots,\ell$. (We have to solve $\ell$ homogeneous linear 
equations in $\delta_1^d,\ldots,\delta_n^d$.) Then
we add a fresh register initialized to $\sum_{t=0}^{p-1}\ket{t}$, and
subtract $\delta_ix$ from the $i$th register. We obtain
$$
\sum_{x=0}^{p-1}\sum_{x_1,\ldots,x_n=0}^{p-1}
\omega^{
\sum_{j=1}^{d}((x_1+\delta_1 x)^j\sum_{k=1}^\ell Y^1_{jk}\hparam_k+\ldots +
(x_n+\delta_n x)^j\sum_{k=1}^\ell Y^n_{jk}\hparam_k)
}\ket{x_1,\ldots,x_n}\ket{x}.$$
Collecting the terms according to the degree of $x$ in the phase,
we can rewrite the state as
$$
\sum_{x=0}^{p-1}\sum_{x_1,\ldots,x_n=0}^{p-1}
\omega^{\sum_{j=0}^dx^j\sum_{k=1}^\ell Z_{jk}(x_1,\ldots,x_n)\hparam_k
}\ket{x_1,\ldots,x_n}\ket{x}.$$
Here $Z_{jk}(x_1,\ldots,x_n)$ is a degree $d-j$ polynomial
in $x_1,\ldots,x_n$. By the choice of $\delta_1,\ldots,\delta_n$,
we have 
$$Z_{dk}(x_1,\ldots,x_n)=
\delta_1^dY^1_{dk}+\ldots+\delta_n^dY^n_{dk}=0.$$
We also have
$$Z_{d-1,k}(x_1,\ldots,x_n)=
d\delta_1^{d-1}Y^1_{dk}x_1+\ldots+d\delta_n^{d-1}Y^n_{dk}x_n
+
\delta_1^{d-1}Y^1_{d-1,k}+\ldots+\delta_n^{d-1,k}Y^n_{dk}.
$$
We have $\delta_i\neq 0$, for at least one index $i$ from
$1,\ldots,n$. As $Y^i_{dk}$ is nonzero for at least one $k$,
the polynomial $Z_{d-1,k}$ contains the term $x_i$ with nonzero
coefficient. Hence, for a random choice of $x_1,\ldots,x_n$, 
it will be nonzero
with probability at least $\frac{p-1}{p}$. Therefore, if we
measure the first $n$ registers, we obtain a state of
the form
$$
\sum_{x=0}^{p-1}
\omega^{\sum_{j=0}^{d-1}x^j\sum_{k=1}^\ell Z_{jk}\hparam_k
}\ket{x},$$
where not all the vectors $Z_{jk}$ are zero.

Starting with $n^{d-1}$ states with degree $d$ phase
(coming from $n^{d-1}$ level set states),
applying this procedure to groups of size $n$ 
we obtain $n^{d-2}$ states with degree $d-1$ phase,
from which we can produce $n^{d-3}$ degree $d-2$ states
and so on. Eventually, with overall failure probability 
at most $n^d/p$, we obtain a state of the form
$$\sum_{x=0}^{p-1}\omega^{x\sum_{k=1}^\ell z_k\hparam_k}\ket{x},$$
with known $z_1,\ldots,z_k$, not all zero.
Applying the inverse Fourier transform of $\F_p$, we
obtain the value for $\sum_{k=1}^\ell z_k\hparam_k$, that is,
a linear equation
for $\hparam_1,\ldots,\hparam_\ell$. Using this equation, we can
substitute a linear combination of the others (and a constant term)
into
one of the parameters, and we can do a recursion with 
$\ell-1$ unknown parameters. 

The whole procedure uses $\ell n^{d-1}$ level set superpositions,
has overall failure probability $\ell n^{d-1}/p$ and requires
$\poly (\ell n^{d-1}\log p)$ time to determine the hidden
coefficients $w_j$. For our task, we take $\ell=md$.

\end{proof}

\section*{Acknowledgements}
The authors are grateful to the anonymous referees for their helpful
remarks and suggestions. The research is partially funded by the Singapore Ministry of Education and the 
National Research Foundation, also through the Tier 3 Grant ``Random numbers 
from quantum processes,'' MOE2012-T3-1-009.
Research also partially supported by the European Commission
IST STREP project Quantum Algorithms (QALGO) 600700,
by the French ANR Blanc program under contract ANR-12-BS02-005 (RDAM project), 
and by the Hungarian National Research, Development and Innovation Office --
NKFIH, Grant NK105645.


\appendix

\end{document}